\documentclass[10pt,onecolumn]{article}
\usepackage[a4paper, hmargin={2.8cm, 2.8cm}, vmargin={2.5cm, 2.5cm}]{geometry}

\usepackage[utf8]{inputenc}
\usepackage{tikz}

\usepackage{amsmath}
\usepackage{amssymb}  
\usepackage{amsthm}
\usepackage{amsfonts}
\usepackage{graphicx}
\usepackage{cancel}
\usepackage{bbm}
\usepackage{url}
\usepackage{enumerate} 
\usepackage{ upgreek }
\usepackage{dsfont}
\usepackage{color}
\usepackage{subcaption}
\usepackage{comment}

\usepackage{makecell}
\usepackage{diagbox}
\usepackage{wrapfig}
\usepackage{rotating}
\usepackage{tabularx}

\usepackage{hyperref}
\hypersetup{
    colorlinks=true,
    linkcolor=blue,
    citecolor=red,
    filecolor=magenta,      
    urlcolor=cyan,
}

\newcommand{\tr}{\textnormal{Tr}}

\renewcommand{\d}[1]{\ensuremath{\textnormal{d}#1}}

\newcommand{\id}{\ensuremath{\mathds{1}}}

\newcommand{\cB}{\mathcal{B}}

\newcommand{\cH}{\mathcal{H}}

\newcommand{\cK}{\mathcal{K}}

\def\beq{\begin{equation}}
\def\eeq{\end{equation}}
\def\bq{\begin{quote}}
\def\eq{\end{quote}}
\def\ben{\begin{enumerate}}
\def\een{\end{enumerate}}
\def\bit{\begin{itemize}}
\def\eit{\end{itemize}}

\def\r|{\right|}

\def\one{\id}
\def\bu{{\bf u}}
\def\bv{{\bf v}}
\def\bx{{\bf x}}
\def\by{{\bf y}}
\def\bw{{\bf w}}

\newcommand\C{\mathbbm{C}}

\theoremstyle{plain}
\newtheorem{thm}{Theorem}[section]
\newtheorem{lem}[thm]{Lemma}

\theoremstyle{definition}

\newtheorem{rem}{Remark}[section]

\renewcommand{\leq}{\leqslant}
\renewcommand{\geq}{\geqslant}

\title{On an inequality of Lin, Kim and Hsieh and Strong Subadditivity}
\author{Eric A. Carlen$^{1}$ and Michael P. Loss$^{2}$\\
\\
\small{$^{1}$ Department of Mathematics, Hill Center,}\\[-6pt]
\small{Rutgers University, 110 Frelinghuysen Road Piscataway NJ 08854-8019 USA}\\[-6pt]\\
\small{$^{2}$ School of Mathematics, Georgia Tech, Atlanta GA 30332}\\[-6pt]
 }

\begin{document}

\maketitle

\footnotetext [1]{Work partially supported by U.S.
National Science Foundation grant  DMS-2055282.}

\footnotetext [2]{Work partially supported by U.S. National
Science Foundation
grant DMS-2154340.}

\begin{abstract}
We give an elementary proof of an inequality of  Lin, Kim and Hsieh that implies strong subadditivity of the non Neumann entropy.
\end{abstract}

\section{Introduction}

Let $\cH_1, \cH_2$ and $\cH_3$ be finite dimensional Hilbert spaces.   For a density matrix $\rho_{123}$ on 
$\cH_1\otimes\cH_2\otimes \cH_3$, 
let $\rho_{12} := \tr_3[\rho_{123}]$, the partial trace of $\rho_{123}$ over $\cH_3$, and $\rho_{1} := \tr_2[\rho_{12}] = \tr_{23}[\rho_{123}]$, etc. In 1968, Lanford and Robinson \cite{LR} conjectured that the von Neumann entropies of $\rho_{123}$, $\rho_2$, $\rho_{12}$ and $\rho_{23}$ are related by
\begin{equation}\label{SSA}
S(\rho_{12}) + S(\rho_{23}) \geq S(\rho_{123}) + S(\rho_2)\ .
\end{equation}
In 1970, while the conjecture was still open, Araki and Lieb \cite{AL} showed by an elementary purification argument that \eqref{SSA} is valid for all tripartite density matrices $\rho_{123}$ if and only if
\begin{equation}\label{LKH3}
S( \rho_{12}) + S(\rho_{23}) - S(\rho_1) - S(\rho_3) \geq 0
\end{equation}
holds for all tripartite density matrices $\rho_{123}$.  They did prove a weaker form of \eqref{SSA}, with 
$S(\rho_2)$ replaced by $\log(\tr[\rho_2^2])$, and showed that this sufficed to prove the thermodynamic limit that Lanford and Robinson had sought to prove. 
However, \eqref{SSA} is of interest for many reasons, and work towards proving it continued. 

To see that \eqref{SSA} and \eqref{LKH3} are equivalent, let $\cH_4$ be a fourth Hilbert space, and let $\Psi$ be a unit vector in $\cH_1\otimes \cH_2\otimes \cH_3\otimes \cH_4$ such that 
$$
\tr_4[|\Psi\rangle\langle\Psi|] = \rho_{123}\ .
$$
Such a unit vector $\Psi$ always exists. By a 1935 theorem of Schr\"odinger \cite{Scr35}, if $\cH$ and $\cK$ are two finite dimensional Hilbert spaces,
and $\Phi$ is any unit vector in $\cH\otimes \cK$, then there are orthonormal sets $\{\bu_1,\dots,\bu_r\}$ and $\{\bv_1,\dots,\bv_r\}$ in $\cH$ and $\cK$ respectively, and positive numbers $\lambda_1,\dots,\lambda_r$ with $\sum_{j=1}^r\lambda_j =1$ such that
\begin{equation}\label{Schmidt}
\Phi := \sum_{j=1}^r \lambda_j^{1/2} \bu_j\otimes \bv_j\ .
\end{equation}
The bi-orthogonal decomposition \eqref{Schmidt} is usually called the {\em Schmidt decomposition}. What Schmidt actually wrote 
about \cite{Sc07} was the closely related singular value decomposition for integral operators. 

Taking the two partial traces of the pure state $|\Phi\rangle\langle\Phi|$ yields
\begin{equation}\label{Sc}
\rho_\cH := \tr_\cK[|\Phi\rangle\langle\Phi|] = \sum_{j=1}^r \lambda_j |\bu_j\rangle\langle \bu_j| \quad{\rm and}\quad 
\rho_\cK := \tr_\cH[|\Phi\rangle\langle\Phi|] = \sum_{j=1}^r \lambda_j |\bv_j\rangle\langle \bv_j|\ .
\end{equation}
Thus, as observed in \cite{AL}, $\rho_\cH$ and $\rho_\cK$ have the same non-zero spectrum, including multiplicities. In particular, 
\begin{equation}\label{EqEnt}
S(\rho_\cH) = S(\rho_\cK)\ .
\end{equation}

To ``purify'' $\rho_{123}$, let  $\rho_{123} =
\sum_{j=1}^r\lambda_j |\bu_j\rangle\langle \bu_j|$ be a spectral decomposition of $\rho_{123}$, and let $\{\bv_1,\dots,\bv_r\}$ be an orthonormal set in 
a Hilbert space $\cH_4$ of dimension at least $r$. By \eqref{Sc},  
 $\Psi = \sum_{j=1}^r \lambda_j^{1/2}\bu_j\otimes \bv_j$ has the  property that $\tr_4[|\Psi\rangle\langle\Psi|] = \rho_{123}$. 
 Now one can consider $\cH_1\otimes\cH_2\otimes \cH_3\otimes \cH_4$ as a bipartite space $\cH\otimes \cK$ in various ways. By \eqref{Sc} we have
 \begin{equation}\label{IDENT}
 S(\rho_{123}) = S(\rho_4)\quad{\rm and}\quad S(\rho_{23}) = S(\rho_{14})\ .
 \end{equation}
 Thus, \eqref{SSA} becomes
 \begin{equation}\label{SSAV}
S(\rho_{12}) + S(\rho_{14}) \geq S(\rho_{4}) + S(\rho_2)\ ,
\end{equation}
which, up to relabeling, is the same as \eqref{LKH3}. Conversely, because of \eqref{IDENT}, once one has proved \eqref{SSAV}, one has proved \eqref{SSA}, as noted in \cite{AL}.  However, until the recent paper  \cite{LKH23}   there was no direct proof of \eqref{SSAV}.  The proof of \eqref{SSA} was given by Lieb and Ruskai \cite{LR} in 1973. Their proof used a deep  convexity result of Lieb, proved in the same year \cite{L73}, that resolved another famous conjecture, the Wigner-Yanase-Dyson Conjecture.  The following, rather surprising theorem was proved in \cite{LKH23}.

\begin{thm}[Lin, Kim and Hsieh]\label{LKHTh}  Let $\rho_{12}$ be an invertible density matrix on $\cH_1\otimes\cH_2$, and  let $\sigma_{23}$ be any density matrix on 
$\cH_2\otimes\cH_3$. Then
\begin{equation}\label{LKH1}
 \rho_1^{-1}\otimes \sigma_{23} \leq \rho_{12}^{-1}\otimes \sigma_3 \ .
\end{equation}
\end{thm}

As a direct corollary, one obtains \eqref{LKH3}, and hence strong subadditivity. 
Since the logarithm is operator monotone,
\begin{equation}\label{LKH2}
\log \rho_{12} + \log \sigma_{23} - \log \rho_1 - \log \sigma_3 \leq 0\ .
\end{equation}
Now let $\rho_{123}$ be a density matrix on $\cH_1\otimes \cH_2\otimes \cH_3$, and apply \eqref{LKH2} with $\rho_{12}= \tr_3[\rho_{123}]$ and $\sigma_{23} = \tr_1[\rho_{123}] = \rho_{23}$. Then taking the trace against $-\rho_{123}$ yields \eqref{LKH3}. The condition that $\rho_{12}$ be invertible is removed by a simple continuity argument \cite{AF04}.

The elegant proof of \eqref{LKH2} in \cite{LKH23} uses a ``twisted'' pair of Stinespring factorizations of a completely positive map to prove that 
$$
\rho_{12}^{1/2}\rho_1^{-1/2} (\sigma_{23}^{1/2}\sigma_3^{-1/2})^* = \rho_{12}^{1/2} \sigma_3^{-1/2} \rho_1^{-1/2} \sigma_{23}^{1/2}
$$
is a contraction on $\cH_1\otimes \cH_2\otimes \cH_3$, and this is equivalent to the statement that 
$$
(\rho_1^{-1} \sigma_{23})^{1/2} \rho_{12} \sigma_3^{-1} (\rho_1^{-1} \sigma_{23})^{1/2} \leq \one\ ,
$$
which in turn is equivalent to \eqref{LKH1}. In this note we give an elementary direct proof of \eqref{LKH1} using little more than the Cauchy-Schwarz inequality and Schmidt decompositions.

\begin{rem}\label{Rem1} The inequality \eqref{LKH1} extends to an inequality for positive operators $X$ on $\cH_1\otimes \cH_2$ and $Y$ on $\cH_2\otimes \cH_3$ since
we may normalize these to form $\rho_{12} = (\tr[X])^{-1}X$ and $\sigma_{23} :=  (\tr[Y])^{-1}Y$, and the normalization factors may be cancelled from both sides.  The normalization is only relevant to the application. 

Next, while there is no restriction on the finite dimensions of the Hilbert spaces $d_1$, $d_2$ and $d_3$, we may freely assume that 
$d_2 \geq \max\{d_1,\d_3\}$.  To see this, let $\widetilde \cH_2$ be a Hilbert space that contains $\cH_2$ as a proper subspace. 
Let $P$ be the projection onto $\cH_2$ in $\widetilde \cH_2$. If $\rho_{12}$ is invertible on $\cH_1\otimes \cH_2$, $\rho_{12}(\epsilon) :=\rho_{12}+\epsilon \one_1\otimes P$ is invertible on
$\cH_1\otimes \widetilde \cH_2$ for all $\epsilon>0$.  Of  course, any density matrix $\sigma_{23}$ on $\cH_2\otimes \cH_3$ may be regarded as a density matrix on $\widetilde \cH_2\otimes \cH_3$. Then the restriction of $\rho_{12}^{-1}(\epsilon)\otimes \sigma_3 \geq \rho_1^{-1}(\epsilon)\otimes \sigma_{23}$
to $\cH_1\otimes \cH_2\otimes \cH_3$ reduces to \eqref{LKH1} as $\epsilon$ decreases to zero. 
\end{rem}

\section{Proof of the inequality of Lin, Kim and Hsieh}

\begin{lem}\label{key} Let $\Psi$ and $\Phi$ be unit vectors in $\cH_1\otimes \cH_2$ and $\cH_2\otimes \cH_3$ respectively, and let $\rho_{12} = 
|\Psi \rangle\langle \Psi|$ and $\sigma_{23} = |\Phi \rangle\langle \Phi|$. Suppose that $\sigma_3$ is invertible. 
Let  $P$ denote the projection into the orthogonal complement of $\Phi$ in $\cH_2\otimes \cH_3$. For $\epsilon>0$, define
\begin{equation}\label{LKH4}
X_{23}(\epsilon) = |\Phi \rangle\langle \Phi| + \epsilon  P\ .
\end{equation}
Then   for all $\epsilon$ sufficiently small, 
\begin{equation}\label{LKH5}
\rho_{12}\otimes \sigma_3^{-1} \leq \left(1 +\sqrt{\epsilon}\right)\rho_1\otimes X_{23}(\epsilon)^{-1}\ .
\end{equation}
\end{lem}

\begin{rem} Since $\sigma_2$ and $\sigma_3$ have the same rank, $\sigma_3$ can only be invertible if $d_2 \geq d_3$. 
\end{rem}

\begin{proof}[Proof of Theorem~\ref{LKHTh}] By Remark~\ref{Rem1}, we may assume $d_2 \geq d_3$.
Since both sides of \eqref{LKH5} are linear in $\rho_{12}$, and since any density matrix is a convex combination of pure states, \eqref{LKH5} is valid for all density matrices $\rho_{12}$ on $\cH_1\otimes \cH_2$. Then by the monotonicity of the inverse function, if $\rho_{12}$ is invertible,
\begin{equation}\label{LKH6}
\rho_{12}^{-1}\otimes \sigma_3 \geq \left(1 + \sqrt{\epsilon}\right)^{-1}\rho_1^{-1}\otimes X_{23}(\epsilon)\ .
\end{equation}
Now taking $\epsilon$ to zero, we recover \eqref{LKH1} for pure states $\sigma_{23}$ such that $\sigma_3$ is invertible. Such pure states are dense by \eqref{Schmidt} and \eqref{Sc}, and again, since any density matrix on $\cH_2\otimes \cH_3$ is a convex combination of pure states, it follows that \eqref{LKH1} is valid for all density matrices $\sigma_{23}$. 
\end{proof}

\begin{proof}[Proof of Lemma~\ref{key}]   Let $\{\Phi_j\}_{1 \leq j \leq d_2d_3}$ be an orthonormal basis for $\cH_2\otimes \cH_3$ in which $\Phi_1 = \Phi$. 
The general unit vector $\varphi\in \cH_1\otimes\cH_2\otimes \cH_3$ has the form $\varphi = \sum_{\ell =1}^{d_2d_3} \bw_\ell\otimes \Phi_\ell$ where
$\sum_{\ell =1}^{d_2d_3} \|\bw_\ell\|^2 =1$.  Then with $B := \rho_1\otimes X_{23}(\epsilon)^{-1}  $,
\begin{equation}\label{LKH8}
\langle \varphi \ ,\ B \varphi\rangle =  \langle\bw_1,\rho_1\bw_1\rangle +
\frac{1}{\epsilon} \sum_{\ell=2}^{d_2d_3}  \langle\bw_\ell,\rho_1\bw_\ell\rangle\ .
\end{equation}

We must compare this with $\langle \varphi, A\varphi\rangle$, where  $A := \rho_{12}\otimes \sigma_3^{-1}$. Recall that for all $a,b,\delta>0$,
$ab  \leq \frac 12(\delta a^2 + \delta^{-1}b^2)$. Hence for
 any $\delta>0$,
\begin{eqnarray}\label{LKH62}
\langle \varphi, A\varphi\rangle &=& \sum_{\ell,\ell'=1}^{d_2d_3} \langle \bw_\ell\otimes \Phi_\ell, A \bw_{\ell'}\otimes\Phi_{\ell'}\rangle\nonumber\\
&=& \langle \bw_1\otimes \Phi, A \bw_{1}\otimes\Phi\rangle + 2\sum_{\ell=2}^{d_2d_3}{\rm Re} \langle \bw_1\otimes \Phi, A \bw_{\ell}\otimes\Phi_{\ell}\rangle
+  \sum_{\ell,\ell'=2}^{d_2d_3} \langle \bw_\ell\otimes \Phi_\ell, A \bw_{\ell'}\otimes\Phi_{\ell'}\rangle\nonumber\\
&\leq& (1+\delta d_2d_3)\langle \bw_1\otimes \Phi, A \bw_{1}\otimes\Phi\rangle + \left(\frac{1}{\delta} +d_2d_3\right)
\sum_{\ell=2}^{d_2d_3}\langle \bw_\ell\otimes \Phi_\ell, A \bw_{\ell}\otimes\Phi_{\ell}\rangle\ .
\end{eqnarray}
Next, with $\mu$ denoting the least eigenvalue of $\sigma_3$,  $A\leq \mu^{-1}\rho_{12}\otimes \one$. Therefore,
\begin{eqnarray*}
\langle \bw_\ell\otimes \Phi_\ell, A \bw_{\ell}\otimes\Phi_{\ell}\rangle &\leq& 
\mu^{-1} \langle \bw_\ell\otimes \Phi_\ell, \rho_{12}\otimes \one \bw_{\ell}\otimes\Phi_{\ell}\rangle\\
&\leq& \mu^{-1} \sum_{\ell' =1}^{d_2d_3} \langle \bw_\ell\otimes \Phi_{\ell'}, \rho_{12}\otimes \one \bw_{\ell}\otimes\Phi_{\ell'}\rangle \leq 
\mu^{-1}\langle \bw_\ell,\tr_{23}[\rho_{12}\otimes \one]\bw_\ell\rangle \\
&=& d_3 \mu^{-1}\langle \bw_\ell,\rho_1\bw_\ell\rangle\ ,
\end{eqnarray*}
where final inequality comes from $\langle \bw_\ell,\tr_{23}[\rho_{12}\otimes \one]\bw_\ell\rangle =\sum_{\ell' =1}^{d_2d_3}
 \langle \bw_\ell\otimes \Phi_{\ell'}, \rho_{12}\otimes \one \bw_{\ell}\otimes\Phi_{\ell'}\rangle$.
Combining this with \eqref{LKH62} yields
\begin{equation}\label{LKH63}
\langle \varphi, A\varphi\rangle \leq (1+\delta d_2d_3)  \langle \bw_1\otimes \Phi, A \bw_{1}\otimes\Phi\rangle + 
d_3\mu^{-1}  \left(\frac{1}{\delta} +d_2d_3\right) \sum_{\ell=2}^{d_2d_3}\langle \bw_\ell,\rho_1\bw_\ell\rangle\ .
\end{equation}

We now claim that
\begin{equation}\label{LKH61}
\langle \bw_1\otimes \Phi, A \bw_{1}\otimes\Phi\rangle \leq \langle \bw_1,\rho_1\bw_1\rangle\ .
\end{equation}
Accepting this, choose $\delta = \sqrt{\epsilon}/ d_2d_3$ and then \eqref{LKH63} becomes
\begin{equation}\label{LKH67}
\langle \varphi, A\varphi\rangle \leq  (1+\sqrt{\epsilon})\langle \bw_1,\rho_1\bw_1\rangle + \frac{d_2d_3^2}{\mu}\left(\frac{1}{\sqrt{\epsilon}} +1\right)
\sum_{\ell=2}^{d_2d_3}\langle \bw_\ell,\rho_1\bw_\ell\rangle\ .
\end{equation}
Then from \eqref{LKH8}, 
${\displaystyle
\langle \varphi, A\varphi\rangle \leq  (1+ \sqrt{\epsilon})\langle \varphi, B\varphi\rangle
}$  for all sufficiently small $\epsilon$.

To prove \eqref{LKH61},
let ${\displaystyle \Psi = \sum_{j=1}^{d_1} \lambda_j^{1/2} \bu_j\otimes \bv_j }$ and ${\displaystyle  \Phi = \sum_{k=1}^{d_3} \mu_k^{1/2} \bx_k\otimes \by_k }$
be Schmidt decompositions of $\Psi$ and $\Phi$. 
Then $\rho_1 = \sum_{j=1}^{d_1}\lambda_j|\bu_j\rangle\langle \bu_j|$\ ,  and
${\displaystyle A = \sum_{j,j',k} \lambda_j^{1/2} \lambda_{j'}^{1/2}
\mu_k^ {-1} |\bu_j\otimes \bv_j\otimes\by_k\rangle\langle  \bu_{j'}\otimes \bv_{j'}\otimes\by_k|}$.
Writing out $\langle \bw_1\otimes \Phi, A \bw_1\otimes \Phi\rangle$ using the Schmidt decomposition for $\Phi$ yields the sum
\begin{equation*}
\sum_{i,i'j,j',k} \lambda_j^{1/2} \lambda_{j'}^{1/2}
\mu_k^{-1}\mu_{i}^{1/2}\mu_{i'}^{1/2}  \langle  \bu_{j'}\otimes \bv_{j'}\otimes\by_k, \bw_1\otimes \bx_i\otimes \by_i\rangle 
\langle \bw_1\otimes \bx_{i'}\otimes \by_{i'}, \bu_{j}\otimes \bv_{j}\otimes\by_k\rangle\ 
\end{equation*}
in which
only the terms with $i= i' = k$ can be non-zero, and hence the sum reduces to 
\begin{equation*}
\sum_{j,j',k} \lambda_j^{1/2} \lambda_{j'}^{1/2}
 \langle \bw_1,\bu_j\rangle\langle\bu_{j'},\bw_1\rangle \langle \bv_{j'},\bx_k\rangle \langle\bx_k,\bv_j\rangle 
 \ .
\end{equation*}
 Define the $d_1\times d_1$ matrix $M$ with entries $M_{j',j} :=  \sum_{k=1}^{d_3}\langle \bv_{j'},\bx_k\rangle \langle\bx_k,\bv_j\rangle$ and  the vector ${\bf a}\in \C^{d_1}$ with entries $a_j := 
  \lambda_j^{1/2}  \langle \bw_1,\bu_j\rangle$. Then $0 < M \leq \one$  (with equality on the right if $d_3= d_2$), and therefore,
$\langle \bw_1\otimes \Phi, A \bw_1\otimes \Phi\rangle = \langle {\bf a},M{\bf a}\rangle \leq \|{\bf a}\|^2 =  \langle \bw_1,\rho_1\bw_1\rangle$.
\end{proof}

The original proof of strong subadditivity \cite{LR} showed that strong subadditivity, the joint convexity of the relative entropy, and the 
concavity of the conditional entropy were all equivalent, and then used the convexity results from \cite{L73} to prove strong subadditivity.  The proof of strong subadditivity using Theorem~\ref{LKHTh} sidesteps deep convexity issues altogether, and indeed, 
using the  equivalence results, provides simple proofs of the joint convexity of the relative entropy, and the 
concavity of the conditional entropy. We thank Trung Nghia Nguyen for pointing out a flaw in an earlier proof of the main lemma.

We close by showing that there are no cases of equality for the operator inequality \eqref{LKH1}. 

\begin{thm}\label{Cases of equality} If ${\rm dim}(\cH_2) > 1$, there are no density matrices $\rho_{12}$ and $\sigma_{23}$ for which there is equality in \eqref{LKH1}. 
\end{thm}

\begin{proof}Suppose there is equality in \eqref{LKH1}. Discarding  the kernel of $\sigma_3$ from $\cH_3$ as needed, we may suppose 
$\sigma_3$ is invertible. (Note that as operators on $\cH_2\otimes \cH_3$, ${\rm ker}(\sigma_3) \subseteq {\rm ker}(\sigma_{23})$). Then with equality in \eqref{LKH1}, $\sigma_{23}$ is in fact invertible, and of course $\rho_1$ is invertible whenever $\rho_{12}$ is invertible. 

Define $X = \sigma_3^{-1/2} \sigma_{23}\sigma_3^{-1/2}$ and $Y = 
  \rho_1^{-1/2}\rho_{12}\rho^{-1/2}$. Then as operators on $\cH_1\otimes\cH_2\otimes \cH_3$,
\begin{equation}\label{LKH30}
X = Y^{-1} 
\end{equation}
Since $X$ acts trivially on $\cH_1$ and $Y$ acts trivially on $\cH_3$, there are $A,B\in \cB^+(\cH_2)$ such that $X = \one\otimes A\otimes  \one$ and 
$Y = \one\otimes B\otimes\one $. 
From the definitions of $X$ and $Y$,  $\tr_{2}[X] = \tr_{2}[Y] = \one$ on $\cH_1\otimes \cH_3$, and this means that $\tr_2[A] = \tr_2[B] =1$. 
However then $\tr_2[B^{-1}] \geq (\tr_2[B])^{-1}$, and equality is only possible if $\cH_2$ is one dimensional.  
Hence equality in \eqref{LKH1} is not possible when $\cH_2$ is non-trivial. 
\end{proof}

\section{Declarations} Conflicts of Interest: There is no conflict of interest. Data Availability: there is no data associated to this paper.


\begin{thebibliography}{99}

{\small 

\bibitem{AF04} R.~Alicki, M.~Fannes, \textit{Continuity of quantum conditional information}, J. Phys. A: Math. Gen. {\bf 37}  (2004), L55-L57.


\bibitem{AL} H.~Araki and  E.H. Lieb, \textit{Entropy Inequalities},
Commun. Math. Phys. {\bf 18}(  1970), 160--170.




\bibitem{LR68} O.~Lanford and D.~Robinson, \textit{Mean Entropy of States in Quantum Statistical Mechanics},  J. Math. Phys. {\bf 9} (1968), 1120.

\bibitem{L73} E.~H.~Lieb, \textit{Convex trace functions and the Wigner-Yanase-Dyson conjecture}, Advances
in Math. {\bf 11} (1973), 267--288.

\bibitem{LR}  E.~H.~Lieb and M.~B.~Ruskai: \textit{Proof of the strong
subadditivity of quantum-mechanical entropy}, J. Math. Phys. {\bf 14} (1973),  1938--1941.



\bibitem{LKH23} T.~C.~Lin, I.~H.~Kim and M.~H.~Hsieh, \textit{A new operator extension of strong subadditivity of quantum entropy},
Lett. Math. Phys. {\bf 113} (2023), 68

\bibitem{Sc07} E.~Schmidt, \textit{Zur Theorie der linearen und nichtlinearen Integralgleichungen
I. Teil: Entwicklung willkürlicher Funktionen nach Systemen vorgeschriebener}, Math. Annalen {\bf 63} (1907), 433-476. 


\bibitem{Scr35} E.~Schr\"odinger,  \textit{Discussion of probability relations between separated systems}, Proc.  Cambridge Phil. Soc. {\bf 31} (1935) 555--563.

}

\end{thebibliography}
\end{document}